\documentclass[openacc]{rstransa}

\usepackage{stmaryrd}
\usepackage{tikz}
\usepackage{tikz-cd}
\usepackage{mathtools}
\usepackage[numbers]{natbib}
\usepackage{hyperref}
\usepackage[mathscr]{euscript}
\usepackage{colonequals}
\usepackage{nicefrac}
\usepackage{floatrow}
\usepackage{blindtext}
\newfloatcommand{capbtabbox}{table}[][\FBwidth]

\usepackage{etoolbox}
\usepackage{algpseudocode}
\usepackage{algorithm}
\algnewcommand{\algorithmicgoto}{\textbf{go to}}%
\algnewcommand{\Goto}[1]{\algorithmicgoto~\ref{#1}}%

\usetikzlibrary{decorations.pathreplacing,calc}

\algnewcommand\algorithmicinput{\textbf{Input:}}
\algnewcommand\INPUT{\item[\algorithmicinput]}
\algnewcommand\algorithmicoutput{\textbf{Output:}}
\algnewcommand\OUTPUT{\item[\algorithmicoutput]}

\usepackage[colorinlistoftodos,prependcaption,textsize=tiny]{todonotes}
\usepackage{soul}

\usepackage{nameref}
\usepackage{textcomp}
\usepackage{breakurl}           
\usepackage{lipsum}
\usepackage{amsmath, amssymb, amsthm, amscd, mathtools}
\usepackage{tikz}
\usepackage{tikz-cd}
\usetikzlibrary{shapes, arrows}
\usetikzlibrary{calc}
\usepackage{xcolor}
\usepackage[toc,page]{appendix}
\usepackage{graphicx}
\usepackage{authblk}
\usepackage{multicol}
\usepackage{blindtext}
\usepackage{enumitem}

\newcommand{\Set}{\textbf{Set}}

\newcommand{\Pos}{\textbf{Pos}}

\newcommand{\supp}{\mathsf{supp}}
\newcommand{\tuple}[1]{\left<#1\right>}

\newcommand{\x}{\mathbf{x}}
\newcommand{\y}{\mathbf{y}}
\newcommand{\s}{\mathbf{s}}

\newcommand{\vi}{\mathbf{v}}

\newcommand{\g}{\mathbf{g}}

\newcommand{\D}{\mathcal{D}}
\newcommand{\E}{\mathcal{E}}

\renewcommand{\L}{\mathcal{L}}
\newcommand{\M}{\mathcal{M}}

\renewcommand{\P}{\mathcal{P}}

\renewcommand{\S}{\mathcal{S}}

\renewcommand{\tuple}[1]{\left\langle #1 \right\rangle}

\newcommand{\comment}[1]{}

\newcommand{\proj}[2]{{#1^{\downarrow #2}}} 
\newcommand{\projn}[2]{#1_{_{\downarrow #2}}} 

\newcommand{\Models}{\mathbb{M}}

\newcommand{\sch}{\mathsf{scheme}}
\newcommand{\Scheme}[1]{\sch(#1)}
\newcommand{\rel}{\mathsf{rel}}
\newcommand{\relation}[1]{\rel(#1)}

\newcommand{\KEYbf}[1]{\textbf{#1}\index{#1| textbf}}

\newcommand{\Pfin}{\P_{\mathsf{fin}}}

\newcommand{\Theory}{\mathscr{T}}
\newcommand{\Structure}{\mathscr{M}}

\newcommand{\nocontentsline}[3]{}
\newcommand{\tocless}[2]{\bgroup\let\addcontentsline=\nocontentsline#1{#2}\egroup}

\newcommand{\defeq}{\colonequals} 

\newcommand{\LC}{\mathsf{LC}}

\newcommand{\SC}{\mathsf{SC}}

\newcommand{\scenario}{\tuple{X, \M, (O_m)}}

\newcommand{\RR}  {\mathbb{R}}

\newcommand{\BB}  {\mathbb{B}} 

\newtheorem{theorem}{\bf Theorem}[section]

\newtheorem{proposition}{\bf Proposition}[section]

\theoremstyle{definition}
\newtheorem{defn}{\bf Definition}[section]
\newtheorem{exmp}{\bf Example}[section]

\newcommand{\ie}{\textit{i.e.}~}

\jname{rsta}
\Journal{Phil. Trans. R. Soc. A}

\begin{document}

\title{Non-locality, contextuality and valuation algebras: a general theory of disagreement}

\author{
Samson Abramsky$^{1}$, Giovanni Car\`{u}$^{1}$}

\address{$^{1}$Department of Computer Science, University of Oxford, Wolfson Building, Parks Road, Oxford OX1 3QD, U.K.}

\subject{Quantum information, quantum foundations, generic inference}


\corres{Samson Abramsky\\
\email{samson.abramsky@cs.ox.ac.uk}}

\begin{abstract}
We establish a strong link between two apparently unrelated topics: the study of conflicting information in the formal framework of valuation algebras, and the phenomena of non-locality and contextuality. 
In particular, we show that these peculiar features of quantum theory are mathematically equivalent to a general notion of \emph{disagreement} between information sources. This result vastly generalises previously observed connections between contextuality, relational databases, constraint satisfaction problems, and logical paradoxes, and gives further proof that contextual behaviour is not a phenomenon limited to quantum physics, but pervades various domains of mathematics and computer science. The connection allows to translate theorems, methods and algorithms from one field to the other, and paves the way for the application of generic inference algorithms to study contextuality.

\end{abstract}


\begin{fmtext}
\section{Introduction}
Non-locality and contextuality are characteristic features of quantum theory which have been recently proven to play a crucial role as fundamental resources for quantum information and computation \cite{HowardEtAl:ContextualityMagic, Raussendorf:ContextualityInMBQC}. In 2011, Abramsky and Brandenburger introduced an abstract mathematical framework based on sheaf theory to describe these phenomena in a unified treatment, thereby providing a common general theory for the study of non-locality and contextuality, which had been carried out in a rather concrete, example-driven fashion until then \cite{AbramskyBrandenburger}. This high-level description shows that contextuality is not a feature specific to quantum theory, but rather a general mathematical property. As such, it can be witnessed in many areas of science, even in domains unrelated to quantum physics.

%
%
%
%
%
\end{fmtext}

%
%
\maketitle
\noindent
Notable examples of the sheaf-theoretic notion of contextuality have been found in connection with relational database theory \cite{Abramsky12:databases}, robust constraint satisfaction \cite{AbramskyGottlobKolaitis2013:RobustConstraintSatisfactionAndLHV}, natural language semantics \cite{AbramskySadrzadeh2014:SemanticUnification}, and logical paradoxes \cite{AbramskyEtAl:CCP2015, AbramskyBarbosaCaruPerdrix:AvNTriple, DeSilva2017}. This profusion of instances motivates the search for a general theory of \emph{contextual semantics}, an all-comprehensive approach able to capture the essence and structure of contextual behaviour. 
In this paper, we propose such a general framework, based on the idea of \emph{disagreement} between information sources. 

The concepts of \emph{querying} and \emph{combining}   \emph{pieces of information} are elegantly captured by the theory of \emph{valuation algebras}, introduced by Shenoy in the late 1980s \cite{Shenoy:1989aa, SHENOY:1990aa}. By constructing a natural definition of disagreement in this framework, completely independently of quantum theory, we show that contextuality is  a special case of a situation where information sources agree locally, yet disagree globally. The valuation algebraic formalism captures an extremely wide range of areas of mathematics and computer science -- including relational databases \cite{kohlas1996information}, constraint satisfaction problems \cite{kohlas2000computation}, propositional logic \cite{shenoy1994consistency, KohlasEtAl1999}, and many more -- therefore, not only does the theory developed in this paper naturally specialise to the aforementioned examples of contextuality outside the quantum realm, but it allows this phenomenon to be recognized in a much wider range of fields. The purpose of this paper is to introduce a general vocabulary for contextual behaviour, which can then be used to translate theorems, methods and algorithms from one field to the other. The generality of the valuation algebraic framework makes the scope for potential new results extremely wide, and lays the foundations for the application of efficient generic inference algorithms \cite{Shenoy:1989aa, shafer1991local, kohlas2003information, pouly2008generic, PoulySoftware2010, pouly2012generic} to the problem of detecting contextuality.

\paragraph{Other approaches to contextuality}
A number of other approaches to developing a general theory of contextuality have emerged over the past few years \cite{spekkens2005contextuality,csw2014graphtheoretic,acin2015combinatorial,dzhafarov2016context}; of these, the Contextuality-by-Default approach \cite{dzhafarov2016context} also emphasizes applications of contextuality beyond quantum theory. However, this approach is focussed on probabilistic models. By encompassing a wider range of models, notably possibilistic ones, the sheaf-theoretic approach allows a much larger class of examples to be recognized as exhibiting contextuality.

\paragraph{Outline} We begin, in Section \ref{sec: valuation algebras intro}, by introducing valuation algebras, inference problems, and a few key examples. In Section \ref{sec: disagreement} we introduce a general definition of disagreement and discuss many instances of local agreement and global disagreement. Section \ref{sec: NL C} reviews the sheaf-theoretic definition of non-locality and contextuality, and presents the connection with disagreement. In Section \ref{sec: disagreement and inference problems}, we show that, in many relevant valuation algebras, detecting disagreement is in fact an inference problem, and introduce the concept of complete disagreement. Section \ref{sec: disagreement and possibilistic forms of contextuality} deals with the connection between disagreement and logical forms of contextuality. Finally, we discuss future research paths in Section \ref{sec: conclusion}.

\section{Valuation algebras and generic inference}\label{sec: valuation algebras intro}
\subsection{Basic definitions}
We begin by reviewing the language of \emph{valuation algebras}. 
In the simplest terms, a valuation algebra is a set of pieces of information, or \emph{valuations}, concerning some variables. Each valuation carries information about a subset of the variables, called its \emph{domain}. Valuations can be \emph{combined} together to obtain joint information, or \emph{projected} to focus the available information on specific variables.

\begin{defn}
Let $V$ be a 
set of variables. We denote by $\Pfin(V)$ the set of finite subsets of $V$. A \KEYbf{valuation algebra} over $V$ is a set $\Phi$ equipped with three operations:
\begin{enumerate}
\item Labelling: $d:\Phi\rightarrow \Pfin(V)::\phi\mapsto d(\phi).$\index{valuation algebra!labelling operation of|textbf}
\item Projection: $\downarrow:\Phi\times \Pfin(V)\rightarrow \Phi:: (\phi, S)\mapsto \proj{\phi}{S}$, for all $S\subseteq d(\phi)$, 
\item Combination: $-\otimes-:\Phi\times \Phi\longrightarrow \Phi:: (\phi,\psi)\mapsto \phi\otimes \psi,$
\end{enumerate}
such that the following axioms are satisfied: 
\begin{enumerate}[label=(A\arabic*)]
\item \emph{Commutative Semigroup}: $(\Phi,\otimes)$ is associative and commutative. \label{semigroup}
\item \emph{Projection}: Given $\phi\in\Phi$ and $S\subseteq d(\phi)$, \label{projection}
\[
d\left(\proj{\phi}{S}\right)=S.
\]
\item \emph{Transitivity}: Given $\phi\in\Phi$ and $S\subseteq T \subseteq d(\phi)$, 
\[
\proj{\left(\proj{\phi}{T}\right)}{S}=\proj{\phi}{S}.
\]\label{transitivity}
\item \emph{Domain}: Given $\phi\in\Phi$,
\[
\proj{\phi}{d(\phi)}=\phi.
\]\label{domain}

\item \emph{Labelling}: For all $\phi,\psi\in\Phi$, \label{labelling}
\[
d(\phi\otimes \psi)=d(\phi)\cup d(\psi)
\]
\item\emph{Combination}: For $\phi,\psi\in\Phi$, with $d(\phi)=S$, $d(\psi)=T$ and $U\subseteq V$ such that $S\subseteq U\subseteq S\cup T$, \label{combination}
\[
\proj{(\phi\otimes\psi)}{U}=\phi\otimes\proj{\psi}{U\cap T}
\]
\end{enumerate}
The elements of a valuation algebra are called \textbf{valuations}. A set of valuations is called a \KEYbf{knowledgebase}. 
A set of variables $D\subseteq V$ is called a \KEYbf{domain}. The \textbf{domain of a valuation}\index{valuation!domain of|textbf} $\phi$ is the set $d(\phi)$. 
\end{defn}

Intuitively, a valuation $\phi\in \Phi$ represents information about the possible values of a finite set of variables $d(\phi)=\{x_1,\dots, x_n\}\subseteq V$, which constitutes the domain of $\phi$. For any set of variables $S\subseteq V$, we denote by $\Phi_S:=\{\phi\in\Phi\mid d(\phi)=S\}$ the set of valuations with domain $S$. Thus, $\Phi=\bigcup_{S\subseteq V}\Phi_S$.
The projection operation can be interpreted as the natural process of \emph{focusing} information over a set of variables to the subset relevant for a given problem. Combination, on the other hand, models the way pieces of information can be merged to obtain knowledge on a larger set of variables. With this interpretation, all of the axioms above should be intuitively clear.

Besides axioms \ref{projection}--\ref{combination}, it is often desirable to add some additional postulates, which collectively give rise to the notion of \emph{information algebra}.
\begin{defn}
Let $\Phi$ be a valuation algebra over a set of variables $V$. 
\begin{itemize}
\item We say that $\Phi$ has \textbf{neutral elements}\index{valuation algebra!neutral elements of|textbf} if it satisfies
\begin{enumerate}[label=(A\arabic*)]
\setcounter{enumi}{6}
\item \emph{Commutative monoid}: \label{monoid}
For each $S\subseteq V$, there exists a \emph{neutral element} $e_S\in\Phi_S$ such that 
\[
\phi\otimes e_S=e_S\otimes \phi=\phi
\] 
for all $\phi\in\Phi_S$. Such neutral elements must satisfy the following identity:
\[
e_S\otimes e_T=e_{S\cup T}
\]
for all subsets $S,T\subseteq V$.
\end{enumerate}
\item We say that $\Phi$ has \textbf{null elements}\index{valuation algebra!null elements of|textbf} if it satisfies
\begin{enumerate}[label=(A\arabic*)]
\setcounter{enumi}{7}
\item \emph{Nullity}: For each $S\subseteq V$ there exists a \emph{null element} $z_S\in\Phi_S$ such that \label{nullity}
\[
\phi\otimes z_S=z_S\otimes \phi=z_S.
\]
Moreover, for all $S,T\subseteq V$ such that $S\subseteq T$, we have, for each $\phi\in\Phi_T$, 
\begin{equation}\label{equ: z}
\proj{\phi}{S}=z_S \Longleftrightarrow \phi=z_T.
\end{equation}
\end{enumerate}
\item We say that $\Phi$ is \textbf{idempotent}\index{valuation algebra!idempotent valutation algebra|textbf} if it satisfies
\begin{enumerate}[label=(A\arabic*)]
\setcounter{enumi}{8}
\item \emph{Idempotency}:\label{idempotency}
For all $\phi\in\Phi$ and $S\subseteq d(\phi)$, it holds that
\[
\phi\otimes\proj{\phi}{S}=\phi
\]
\end{enumerate}
\item If $\Phi$ satisfies axioms \ref{monoid}--\ref{idempotency}, it is called an \KEYbf{information algebra}
\end{itemize}
\end{defn}

There are simple intuitions behind these additional axioms. Neutral elements correspond to `useless information', in the sense that they do not improve any other information with which they are combined. Null elements, on the other hand, can be interpreted as destructive information, \ie knowledge that corrupts any other valuation to the point of making it useless. Idempotency is the signature axiom of qualitative or logical, rather than quantitative, e.g.~probabilistic, information. It says that counting \emph{how many times} we have a piece of information is irrelevant.

\subsection{Basic examples}
Information appears in many different ways: news, data, statistics, propositions, etc. The valuation algebraic formalism brings all of these instances within the scope of a joint theory. In this section, we present a few examples, which are particularly meaningful for our study.

We begin by introducing the notion of \emph{frame}. Given a variable $x\in V$, the \textbf{frame} of $x$, denoted by $\Omega_x$, is the set of possible values for $x$. Given a set of variables $S\subseteq V$, we can model the piece of information constituted by variables in $S$ having acquired specific values in their respective frames as a tuple $\x\in
\Omega_S:=\prod_{x\in S}\Omega_x.$
Given a tuple $\x\in\Omega_S$ and a subset $T\subseteq S$, we will denote by $\projn{\x}{T}$ the cartesian projection of $\x$ onto $\Omega_T$.\footnote{We have used the different notation $\projn{(\cdot)}{(-)}$ for the projection of tuples to distinguish it from the projection $\proj{(\cdot)}{(-)}$ of valuation algebras.}

\subsubsection*{Relational databases}\index{relational database|textbf}

Consider the following data table, taken from \cite{Abramsky12:databases}.
\begin{center}
\begin{tabular}{ c | c | c | c }
branch-name & account-no & customer-name & balance\\
\hline
Cambridge & 10991-06284 & Newton & 2,567.53 GBP\\
Hanover & 10992-35671 & Leibniz & 11,245.75 EUR\\
$\dots$ & $\dots$ & $\dots$ & $\dots$ \\
\end{tabular}
\end{center}

\noindent We identify the list of \emph{attributes}\index{relational database!attribute of|textbf}:
$V\defeq\{\text{branch-name, account-no, custmer-name, balance}\}$, which
label the columns of the table, and constitute the variables of the valuation algebra we shall now define.
Each entry of the table is a tuple specifying a value for each of the attributes. Thus, the full table is  a set of tuples, or a \emph{relation} in database  terminology. A relational database consists of a set of such relations \cite{ullman1984principles}.
Abstracting from this example, given a set of attributes $V$, we define a relation over a finite subset $S\subseteq V$ as a subset $R\subseteq \Omega_S$. The domain of $R$, often called \emph{schema} of $R$ in database theory, is then $d(R)=S$. 
Combination is given by the \emph{natural join}\index{natural join|textbf}: given two relations $R_1$ and $R_2$ with domain $S$ and $T$ respectively, then
\[
R_1\otimes R_2:=R_1\Join R_2=\{\x\in\Omega_{S\cup T}\mid \projn{\x}{S}\in R_1~\wedge~\projn{\x}{T}\in R_2\},
\]
which can be easily shown to be idempotent.
Given a relation $R$ with domain $d(R)=S$, and a subset $T\subseteq S$, we define projection as follows:
\begin{equation}\label{equ: databases projection}
\proj{R}{T}:=\{\projn{\x}{T}\mid \x\in R\}.
\end{equation}
For each $S\subseteq V$, we define the neutral element as $e_S:=\Omega_S$. On the other hand, the null element is defined as $z_S:=\emptyset$. The set $\Phi\defeq \bigcup_{S\subseteq V}\Pfin(\Omega_S)$, equipped with the operations and elements above is thus an information algebra.

\subsubsection*{Semiring valuation algebras}\index{valuation algebra!semiring valuation algebra|textbf}

Let $\tuple{R, +, \cdot, 0 ,1}$ be a commutative semiring and $V$ a set of variables. A semiring valuation with domain a finite subset $S\subseteq V$ is a function
\[
\phi: \Omega_S\longrightarrow R.
\]
Let $\Phi\defeq \bigcup_{S\in\Pfin(V)} R^{\Omega_S}$ be the set of all such valuations. We define:
\begin{enumerate}
\item \emph{Labelling}: $d:\Phi\rightarrow \P(V)$, with $d(\phi)=S$ if $\phi\in\Phi_S$.
\item \emph{Combination}: $\otimes:\Phi\times\Phi\rightarrow \Phi$, where, for all $\phi,\psi\in\Phi$ and $\x\in\Omega_{d(\phi)\cup d(\psi)}$, we have
\[
(\phi\otimes \psi)(\x):=\phi\left(\projn{\x}{d(\phi)}\right)\cdot\psi\left(\projn{\x}{d(\psi)}\right).
\]
\item \emph{Projection}: $\downarrow:\Phi\times\P(V)\rightarrow \Phi$, where, for all $\phi\in \Phi$, $T\subseteq d(\phi)$ and $\x\in\Omega_T$, we have
\[
\proj{\phi}{T}(\x):=\sum_{{\substack{\y\in\Omega_S: \\ \projn{\y}{T}=\x}}}\phi(\y).
\]
\end{enumerate}
This valuation algebra is idempotent only when $R$ is idempotent. 
The neutral element $e_S\equiv 1$ is the function that assigns $1$ to each $\x\in\Omega_S$. The null element $z_S\equiv 0$ is the $0$-function. Semiring valuations are often referred to as \textbf{$R$-potentials}. In the case where $R=\mathbb{R}_{\ge 0}$, the corresponding valuation algebra is the one of \emph{probability potentials}\index{valuation algebra!valuation algebra of probability potentials|textbf} \cite{pouly2012generic}.


\subsection{Advanced examples: language and models}
In many situations, information concerns the validity of propositions formulated in a logical language $\L$. This idea can be captured in considerable generality: we suppose $\L$ is simply a set of \emph{sentences} concerning a basic set of variables $V$, regardless of its syntactic structure, and we assume there is a set $\Models$ of possible \emph{models} for such sentences. Given a subset $Q\subseteq V$, we denote by $\L_Q$ the sublanguage of $\L$ which only involves variables in $Q$. Similarly, let $\Models_Q$ be the set of models for sentences in $\L_Q$. For each $Q\subseteq V$ we assume a binary relation $\models_Q\subseteq \Models_Q\times \L_Q$, and we say that $m\in\Models_Q$ is a model for $s\in\L_Q$ if and only if $m\models_Q s$.
We also assume $\Models$ to be equipped with a projection function: for all $Q\subseteq U$,  $$\downarrow :\Models_U\rightarrow \Models_Q:: m\mapsto \projn{m}{Q}.$$ We require this map to behave essentially like cartesian projection \cite{pouly2012generic}.\footnote{More precisely, we assume $\Models$ to be a \emph{tuple system} \cite{pouly2012generic}, which is a generalisation of a cartesian product of sets equipped with the usual cartesian projection. For this reason, the projection of $\Models$ will be denoted exactly like cartesian projection. We shall not give a full list of the axioms of a tuple system as they are not needed in this paper.} This is needed in order to model a situation where an element $m\in\Models_U$ is a model for a sentence $s$ in $\L_Q$ when $Q$ is strictly contained in $U$. More specifically, we assume that, for any $Q\subseteq U$, $m\in\Models_U$ and $s\in\L_Q$, 
\[
m\models_U \iota(s) \Leftrightarrow \projn{m}{Q}\models_Q s,
\]
where $\iota: \L_Q\hookrightarrow\L_U$ is the inclusion function. This condition states that the projections $\downarrow$ and the inclusions $\iota$ constitute an \emph{infomorphism} \cite{pouly2012generic}, which is a common aspect of many instances of language and models, such as propositional and predicate logic.


The following will be useful for some of the examples in Section \ref{sec: disagreement}\ref{local agreement vs global disagreement}.
For all $Q\subseteq V$ and any $S\subseteq \L_Q$, let
\[
\Structure_Q(S):=\{m\in\Models_Q: m\models_Q s,\forall s\in S\}
\]
denote the set of models for sentences in $S$. Similarly, for all $M\subseteq \Models_Q$, let
\[
\Theory_Q(M):=\{s\in\L_Q: m\models_Q s,\forall m\in M\}.
\]
be the $Q$-theory of $M$, \ie the set of sentences in $\L_Q$ satisfied by all models in $M$. 

With this premise, one can define the valuation algebra of \textbf{information sets} over the family $\{\tuple{\L_Q,\Models_Q,\models_Q}\}_{Q\subseteq V}$ as follows. The set of variables is $V$, and the elements of the algebra are sets of models $M\subseteq \Models_Q$, where $Q\subseteq V$ is finite. Let $\Phi\defeq\bigcup_{Q\in\Pfin(V)}\P(\Models_Q)$. The operations are 
\begin{itemize}
\item \emph{Labelling}: Given an information set $M\subseteq\Models_Q$, its label is defined to be $d(M)=Q.$
\item \emph{Combination}: For all $M_1\subseteq \Models_Q$ and $M_2\subseteq \Models_U$, let
\[
M_1\otimes M_2:=M_1\Join M_2=\{m\in\Models_{Q\cup U}: \projn{m}{Q}\in M_1 \wedge \projn{m}{U}\in M_2\}.
\]
\item \emph{Projection}: Given an information set $M$ and a domain $Q\subseteq d(M)$, we define
\[
\proj{M}{Q}:=\{\projn{w}{Q}: w\in M\}.
\]
\end{itemize}
The neutral elements of the algebra are the sets $e_Q=\Models_Q$, while the null elements are $z_Q=\emptyset$. The algebra is clearly idempotent, hence it is an information algebra. 

This general discussion leads to valuation algebraic representations of propositional logic, predicate logic, linear equation systems, systems of linear inequalities, any many other examples \cite{wilson1999logical, pouly2012generic}. One can also show that the algebra of relational databases is also captured by this general setting. For the purposes of this paper, we are particularly interested in the instances of propositional logic and constraint satisfaction problems, which we shall now briefly review.

\subsubsection*{Propositional logic} Suppose the language $\L$ is propositional logic over a countable set $P$ of propositional symbols. Sentences are well-formed propositional formulae with variables in $P$. The models $\Models$ are just truth assignments, \ie maps $v:P\rightarrow\{0,1\}$. Therefore, $\Models_Q$ is comprised of assignments $w:Q\rightarrow\{0,1\}$, while the relation $\models_Q$ is given by the usual semantics of classical propositional logic. 
The projection function of $\Models$ is given by function restriction: any assignment $w:Q\rightarrow\{0,1\}$ can be restricted to $w|_U:U\rightarrow\{0,1\}$ for any $U\subseteq Q$. This allows to define the information algebra of propositional information sets, where combination of $M_1\subseteq \Models_Q$ and $M_2\subseteq \Models_U$ is
\[
M_1\otimes M_2:=\{w\in\Models_{Q\cup U}: w|_Q\in M_1 \wedge w|_U\in M_2\}.
\]
while projection of $M\subseteq \Models_Q$ to a subset $U\subseteq Q$ is given by $\proj{M}{U}:=\{w|_U: w\in M\}.$

\subsubsection*{Constraint satisfaction problems}\index{constraint satisfaction problem|textbf}
Constraint satisfaction problems (CSP) are particularly useful to model various kinds of problems in mathematics and computer science thanks to their general and versatile formulation. A CSP is a triple $\tuple{X, D, C}$, where $X=\{x_1,\dots x_n\}$ is a set of \emph{variables}, $D=\{D_1,\dots D_n\}$ is the set of the respective \emph{domains}\footnote{This shall not be confused with the term \emph{domain} of the valuation algebra formalism. Rather, the domain of a variable in this setting is analogue to the concept of frame in valuation algebras.} of values, and $C=\{c_1,\dots c_m\}$ is a set of \emph{constraints}. Each variable $x_i$ can take values in its domain $D_i$. A constraint $c_i\in C$ is a pair $\tuple{T_i,R_i}$, where $T_i:=\{x_{i_1},\dots,x_{i_k}\}\subseteq X$ and $R_i$ is a $k$-ary relation $R\subseteq\prod_{j=1}^kD_{i_j}$. The set $T_i$ is also called the \emph{scheme} of $c_i$, and it is denoted by $\Scheme{c_i}$, whereas $R_i:=\relation{c_i}$.
An \emph{evaluation} of the variables is a map $\vi: S\longrightarrow \coprod_{x_i\in S}D_i$
that assigns to each variable $x_j\in S$ a value $v(x_j)$ in its domain $D_j$. Such a map can be seen as an element of $
D_S:=\prod_{x_i\in S}D_i.$
For each $S\subseteq X$, we define $\Models_S$ to be the set of all evaluations on $S$, \ie $\Models_S:=D_S$. The projection of models is simple cartesian projection. The associated language $\L_S$ is defined by
\[
\L_S:=\{c_i\in C: \Scheme{c_i}\subseteq S\}. 
\]
We say that an evaluation $\vi$ of the variables in $S$ satisfies a constraint $c=\tuple{T, R}$, if
$\projn{\vi}{S\cap T}\in R$.  An evaluation is called \emph{consistent} when it does not violate any constraint. It is called \emph{complete} if it includes all variables in $X$. It is called a \emph{solution} if it is consistent and complete. 

Satisfiability of a constraint defines a relation $\models_S$ for each subset $S\subseteq X$: given $\vi\in\Models_S$ and $c\in \L_S$, 
\[
\vi\models_S c\Longleftrightarrow S\cap \Scheme{c}=\emptyset \text{ or }\projn{\vi}{S\cap \Scheme{c}}\in\relation{c}.
\]
Thus we obtain a family $\{\tuple{\L_S,\Models_S,\models_S}\}_{S\subseteq X}$. This allows to define the algebra of information sets, where combination between $M_1\subseteq \Models_S$ and $M_2\subseteq \Models_U$ is
\[
M_1\otimes M_2:=\{\vi\in\Models_{Q\cup U}: \projn{\vi}{Q}\in M_1 \wedge \projn{\vi}{U}\in M_2\},
\]
while projection of $M\subseteq \Models_Q$ to a subset $U\subseteq Q$ is given by $\proj{M}{U}:=\{\projn{\vi}{U}: \vi\in M\}.$

\subsection{Inference problems}

Since valuations model pieces of information, we are naturally drawn to formulate the classic problem of extracting relevant knowledge about a given query from a  knowledgebase. In the valuation algebra theory, such a task is called an \emph{inference problem}, and is formally defined as follows:

\begin{defn}
Given a valuation algebra $\Phi$, a knowledgebase $\{\phi_1,\dots,\phi_n\}\subseteq \Phi$, and a domain $D\subseteq d(\phi_1\otimes\cdots\otimes \phi_n)$, we call an \KEYbf{inference problem} the task of computing 
\[
\proj{\phi}{D}=\proj{(\phi_1\otimes\dots\otimes\phi_n)}{D}.
\]
The valuation $\phi=(\phi_1\otimes\dots\otimes\phi_n)$ is called \textbf{joint valuation}\index{inference problem!joint valuation of|textbf} or \textbf{objective function}, while $D$ is called a \textbf{query}. 
\end{defn}

Many  natural questions in mathematics and computer science can be easily reformulated as inference problems. This level of generality has sparked the development of \emph{generic inference} \cite{Shenoy:1989aa, shafer1991local, kohlas2003information, pouly2008generic, PoulySoftware2010, pouly2012generic}, a theory that aims to produce high-level algorithms to solve inference problems, which can then be applied to a wide range of situations. 

\section{A theory of disagreement}\label{sec: disagreement}
Inference problems capture the essence of information as a carrier of knowledge used to answer specific questions. 
However, they do not take into account another fundamental concept related to knowledge: disagreement. We will now propose a general, natural definition of disagreement in the valuation algebraic language. Then, we will show that this notion elegantly encapsulates non-locality and contextuality. 

\subsection{Defining disagreement}
Consider a valuation algebra $\Phi$, on a set of variables $V$. 
A natural way to say that two valuations $\phi$ and $\psi$ in $\Phi$ agree is to say that they provide exactly the same information when restricted to their common variables. By directly translating this idea in symbols, we say that $\phi_i$ and $\phi_j$ \textbf{agree} if
\begin{equation}\label{equ: local agreement}
\proj{\phi}{d(\phi)\cap d(\psi)}=\proj{\psi}{d(\phi)\cap d(\psi)}
\end{equation}
Now, consider a situation where more than two valuations are involved: let $K=\{\phi_1,\dots,\phi_n\}\subseteq \Phi$ be a knowledgebase. One way to generalise the definition above would be to say that the valuations in $K$ agree if they agree pairwise, a condition which we call \textbf{local agreement}. However, local agreement might not be enough to indisputably conclude that $\phi_1,\dots,\phi_n$ agree on everything. Indeed, although the valuations may provide the same information on \emph{common} variables, they do not necessarily share the same \emph{global} information on \emph{all} of the variables. For instance, different people might agree with each other on specific subjects, while not sharing the same \emph{global} opinion on \emph{all} of the subjects.
In order to model these situations, we say that $K$ is a \textbf{globally agreeing} knowledgebase if $\phi_i,\dots,\phi_n$ share a common global opinion, that is if there exists a valuation $\gamma\in \Phi_X$, where $X=\bigcup_{i=1}^n d(\phi_i)$, such that
\begin{equation}\label{equ: agreement}
\proj{\gamma}{d(\phi_i)}=\phi_i,~\forall 1\leq i\leq n. 
\end{equation}
Concretely, this means that the information carried by each individual valuation $\phi_i$ comes as a restriction of a `truth' valuation $\gamma$ which is implicitly agreed upon by all the sources. Given this premise, we say that $\phi_1,\dots,\phi_n$ \textbf{disagree} if such a global valuation $\gamma$ does not exist. 

Notice that, as one would expect, global agreement implies local agreement, indeed, for all $1\leq i,j\leq n$,
\[
\proj{\phi_i}{d_i\cap d_j}\stackrel{\eqref{equ: agreement}}{=}\proj{\left(\proj{\gamma}{d_i}\right)}{d_i\cap d_j}\stackrel{\ref{semigroup}}{=}\proj{\gamma}{d_i\cap d_j}\stackrel{\text{(A4)}}{=}\proj{\left(\proj{\gamma}{d_j}\right)}{d_i\cap d_j}\stackrel{\eqref{equ: agreement}}{=}\proj{\phi_j}{d_i\cap d_j}.
\]
where $d_{i,j}\defeq d(\phi_{i,j})$.

\subsection{Local agreement vs global disagreement}\label{local agreement vs global disagreement}
Local disagreement is generally easy to spot since it arises as a direct contradiction between two information sources. A much more subtle scenario arises when sources agree locally, yet disagree globally.
In this section, we present some interesting real-world examples of this kind of disagreement using different valuation algebras. 

\begin{exmp}
Our first example is inspired by a recent paper by Zadrozny \& Garbayo \cite{zadrozny2018sheaf}, where breast cancer screening guidelines from three different accredited sources are analysed. The following protocols are obtained by slightly tweaking the instruction contained therein. 

\begin{enumerate}
\item \emph{Screening with mammography annually, clinical breast exam annually or biannually}\label{equ: one}
\item \emph{Women aged 50 to 54 years should get mammograms. Women aged 55 years and older should switch to clinical breast exams}\label{equ: two}
\item \emph{Women aged 50 to 54 years should undergo an exam every year. Women aged 55 years and older should be examined every 2 years}\label{equ: three}
\end{enumerate}

We can represent the information provided by these medical bodies as a knowledgebase of the information algebra of relational databases. We have the variables $\{a,e,f\}$ representing age intervals, exam type and exam frequency respectively. The frames for each variable are $\Omega_a=\{54^-, 54^+\}$, $\Omega_e=\{\text{MG}, \text{CBE}\}$, $\Omega_f=\{\text{Y},\text{2Y}\}$. Each source $i=1,2,3$ from the example can be described by 3 relations, defined by
\[
\begin{split}
R_1 &=\left\{\tuple{\text{M}, \text{Y}}, \tuple{\text{CBE}, \text{Y}}, \tuple{\text{CBE}, \text{2Y}}\right\},\\
R_2 &=\left\{\tuple{54^-, \text{M}},\tuple{54^+, \text{CBE}}\right\},\\
R_3 &=\left\{\tuple{54^-, \text{Y}},\tuple{54^+, \text{2Y}}\right\},
\end{split}
\]
with $d(R_1)=\{e,f\}$, $d(R_2)=\{a,e\}$ and $d(R_3)=\{a,f\}$. It is easy to see that all these sources agree locally, \ie they do not directly contradict each other. For instance, $\proj{R_1}{\{e\}}=\{\text{M}, \text{CBE}\}=\proj{R_2}{\{e\}}$. However, we can show that they disagree globally. To see this, we combine $R_1,R_2,R_3$:
\[
G:=R_1\otimes R_2\otimes R_3=\left\{\tuple{\text{M},\text{Y},54^-},\tuple{\text{CBE},\text{2Y},54^+}\right\}.
\]
Later, in Proposition \ref{prop: truth}, we will show that if a truth valuation exists, then $G$ is one of them. Now, if we project $G$ back onto $d(R_1)$, we have
\[
\proj{G}{d(R_1)}=\{\tuple{\text{M}, \text{Y}}, \tuple{\text{CBE}, \text{2Y}}\},
\]
which is different from $R_1$, since $\proj{G}{d(R_1)}$ does not include $\tuple{\text{CBE}, \text{Y}}$. We conclude that a truth valuation does not exist, and thus $R_1,R_2,R_3$ disagree globally, despite satisfying local agreement. This means that it is impossible to follow all the guidelines simultaneously. More specifically, since $\tuple{\text{CBE}, \text{Y}}\notin \proj{G}{d(R_1)}$, it is impossible to follow both guidelines \ref{equ: two} and \ref{equ: three}, if one opts to undergo clinical breast exams annually, as suggested by \ref{equ: one}. 
%
%
%
This is a very simple example of a real situation where information sources disagree in a very subtle way. It is also an example of a database which does not admit a universal relation, a property that has been proved to be equivalent to contextuality in \cite{Abramsky12:databases}.
\end{exmp}

\begin{exmp}\label{exmp: Malawi}
Consider the following problem. We want to colour a political map of the geographical region surrounding Malawi using 3 colours -- say red, green and yellow -- with the condition that adjacent countries should be coloured differently.  A blank map is pictured in Figure \ref{fig: Malawi}. 
\begin{figure}[htbp]
\RawFloats
\centering
\begin{minipage}[c]{.40\textwidth}
\centering
\includegraphics[scale=0.4]{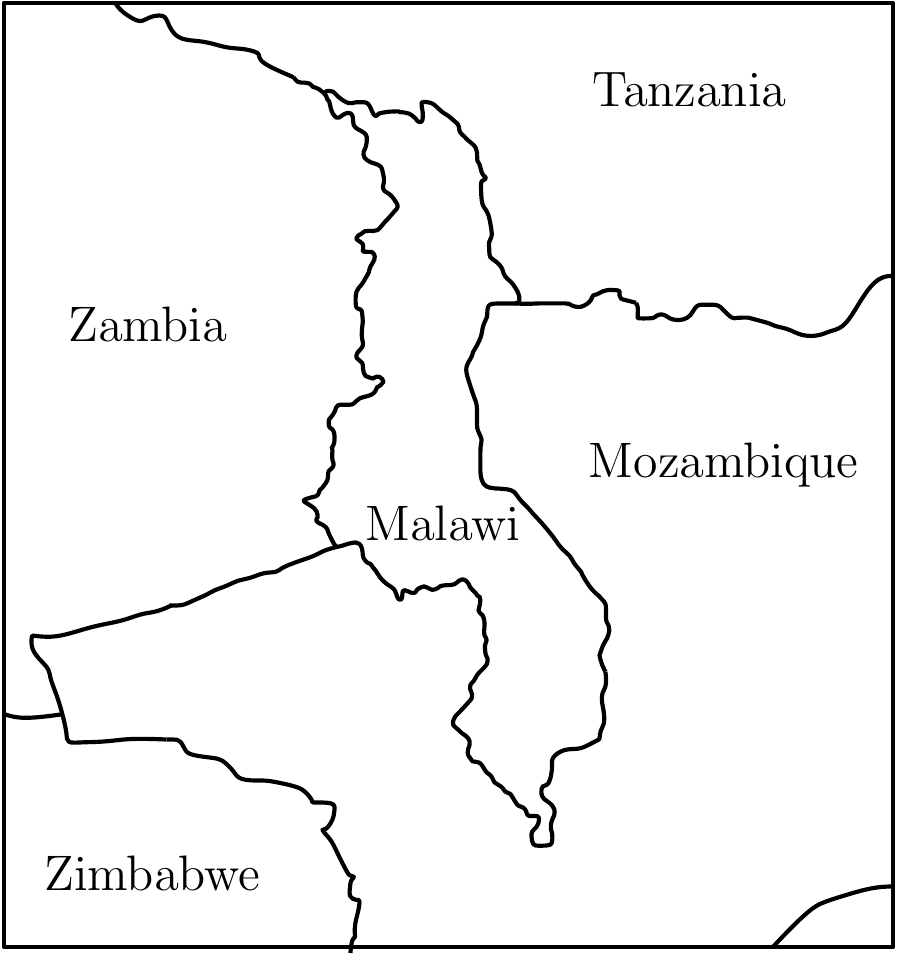}
\caption{A blank map of the geographical region surrounding Malawi.}\label{fig: Malawi}
\end{minipage}%
\hspace{10mm}%
\begin{minipage}[c]{.40\textwidth}
\centering
\includegraphics[scale=0.5]{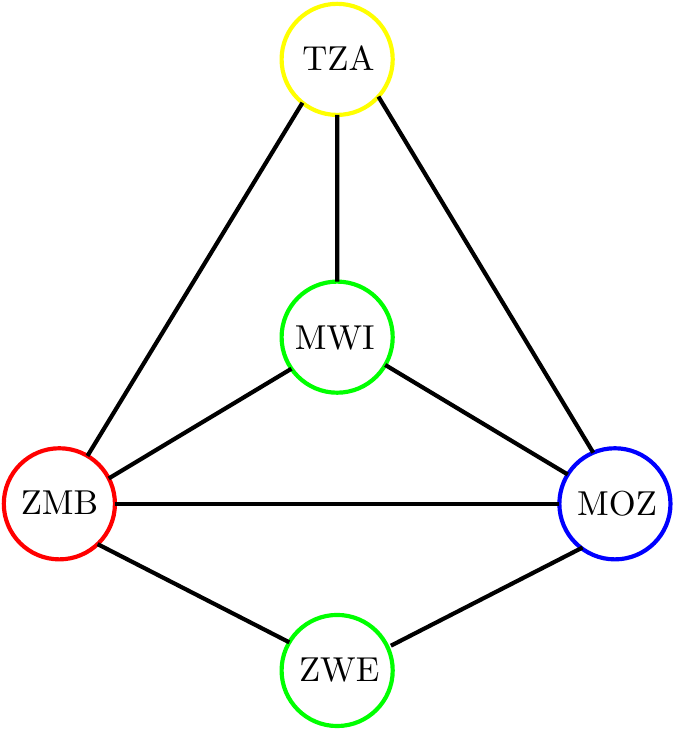}
\caption{The constraint graph of the CSP. A $4$-colouring of the graph is shown. The chromatic number of the graph is $4$.}\label{fig: coloring}
\end{minipage}
\end{figure}

\end{exmp}
We can model this problem as a CSP. The set of variables $V$ is constituted by a variable for each of the 5 countries in the map, \ie $V=\{\mathsf{MOZ}, \mathsf{MWI}, \mathsf{TZA}, \mathsf{ZMB},\mathsf{ZWE}\}$
The domain for each variable is the set of colours we can attribute to the country, \ie $D=\{g,r,y\}$, and it is the same for each variable. Let $S:=\{\tuple{g,r},\tuple{g,y},\tuple{r,y}\}.$
There are $8$ constraints $c_i$, $1\leq i\leq 8$ defined by $\relation{c_i}=S$ for all $1\leq i\leq 8$ and 
\begin{gather*}
T_1=\{\mathsf{MOZ}, \mathsf{MWI}\}, T_2=\{\mathsf{MOZ}, \mathsf{TZA}\}, T_3=\{\mathsf{MOZ}, \mathsf{ZMB}\}, T_4=\{\mathsf{MOZ}, \mathsf{ZWE}\},\\ T_5=\{\mathsf{MWI},\mathsf{TZA}\}, T_6=\{\mathsf{MWI},\mathsf{ZWE}\}, T_7=\{\mathsf{TZA},\mathsf{ZMB}\}, T_8=\{\mathsf{ZMB},\mathsf{ZWE}\},
\end{gather*}
where $T_i=\Scheme{c_i}$ for all $1\leq i\leq 8$. Consider the valuation algebra of information sets for CSPs. The knowledgebase for this problem is comprised of $8$ valuations $\{\phi_i\}_{i=1}^8$ defined by
\begin{gather*}
\phi_i:=\M_{T_i}(C)=\{\vi\in D_{T_i}: \vi\models_{T_i}c_j,~\forall 1\leq j\leq 8\}.
\end{gather*}
For instance, $\phi_1=\{\tuple{\mathsf{MOZ}, \mathsf{MWI}}\mapsto\tuple{g,r}, \text{ or } \tuple{g,y}, \text{ or }\tuple{r,y}\}.$
It is easy to see that all the valuations agree locally, indeed, for all $1\leq i,j\leq 8$, we either have $T_i\cap T_j=\emptyset$, in which case local agreement is trivially satisfied, or $T_i\cap T_j$ has exactly one element $t$, and we have
\[
\proj{\phi_i}{T_i\cap T_j}=\{t\mapsto g \text{ or } r \text{ or } y\}=\proj{\phi_j}{T_i\cap T_j}.
\]
However, the $\phi_i$'s disagree globally. Indeed, if there was a valuation $\gamma$ which projects onto each $\phi_i$, it would imply that there is a solution to the CSP, and thus a colouring of the map using only three colours. It is easy to see that this is not the case by simply looking at the \emph{constraint graph} (cf. Figure \ref{fig: coloring}) of the problem, and conclude that its chromatic number is $4$. 


\begin{exmp}\label{exmp: liar}
To cite an example from logic, consider the the famous liar's paradox, whose structure has been linked to contextuality in \cite{AbramskyEtAl:CCP2015}. The standard paradox consists of the sentence:
\[
S: S \text{ is false}. 
\]
We can generalise this to a \emph{liar cycle} of length $n$, \ie a sequence of statements:
\[
S_1 : S_2 \text{ is true}, \quad\quad S_2 : S_3 \text{ is true},\quad\quad\dots\quad\quad S_{n-1} : S_n \text{ is true},\quad\quad S_n : S_1 \text{ is false}.
\]
These statements can be modelled as a series of formulae in propositional logic. Let $V:=\{s_1,\dots,s_n\}$ be a set of variables, each representing one of the statements $S_i$ above. The $n$ liar cycle can be rewritten as follows:
\begin{equation}\label{equ: liar}
s_1 \leftrightarrow s_2, \quad\quad s_2 \leftrightarrow s_3,\quad\quad\dots\quad\quad s_{n-1} \leftrightarrow s_n,\quad\quad s_n \leftrightarrow \neg s_1.
\end{equation}
We define the following valuations: 
\[
\phi_n:=\M_{\{s_1,s_n\}}(\{s_n\leftrightarrow\neg s_1\})=\{\vi:: \tuple{s_1,s_n}\mapsto \tuple{1,0} \text{ or } \tuple{0,1}\},
\]
and for all $1\leq i\leq n-1$,
\[
\phi_i:=\M_{\{s_i,s_{i+1}\}}(\{s_i\leftrightarrow s_{i+1}\})=\{\vi::\tuple{s_i,s_{i+1}}\mapsto \tuple{0,0} \text{ or }\tuple{1,1}\}.
\]
It is easy to see that the $\phi_i$'s agree locally. Indeed, both $0$ and $1$ are valid assignments for a single variable, which is the most two valuations can have in common. 
On the other hand, the fact that the liar cycle gives rise to a paradox means that the valuations do not agree globally. Indeed, a global assignment of truth values to each $s_1,\dots, s_n$ consistent with formulae \eqref{equ: liar} is impossible, as they collectively yield $s_1\leftrightarrow\neg s_1$. 
\end{exmp}

\section{Non-locality and contextuality}\label{sec: NL C}

We shall now expose a surprising connection between the valuation algebraic theory of disagreement with a completely different topic, namely the study of non-locality and contextuality in quantum foundations. This will show that all of the examples of local agreement vs global disagreement we have just presented are in fact mathematically equivalent to contextuality. We will adopt the sheaf-theoretic description of contextuality developed in \cite{AbramskyBrandenburger}. We will now review the main definition, assuming basic knowledge of sheaf theory and category theory.

\subsection{The sheaf-theoretic structure of non-locality and contextuality}
Our starting point is a simple idealised experiment, depicted in Figure \ref{fig: fig1}.
\begin{figure}[htbp]
\begin{floatrow}
\ffigbox{%
\includegraphics[scale=0.4]{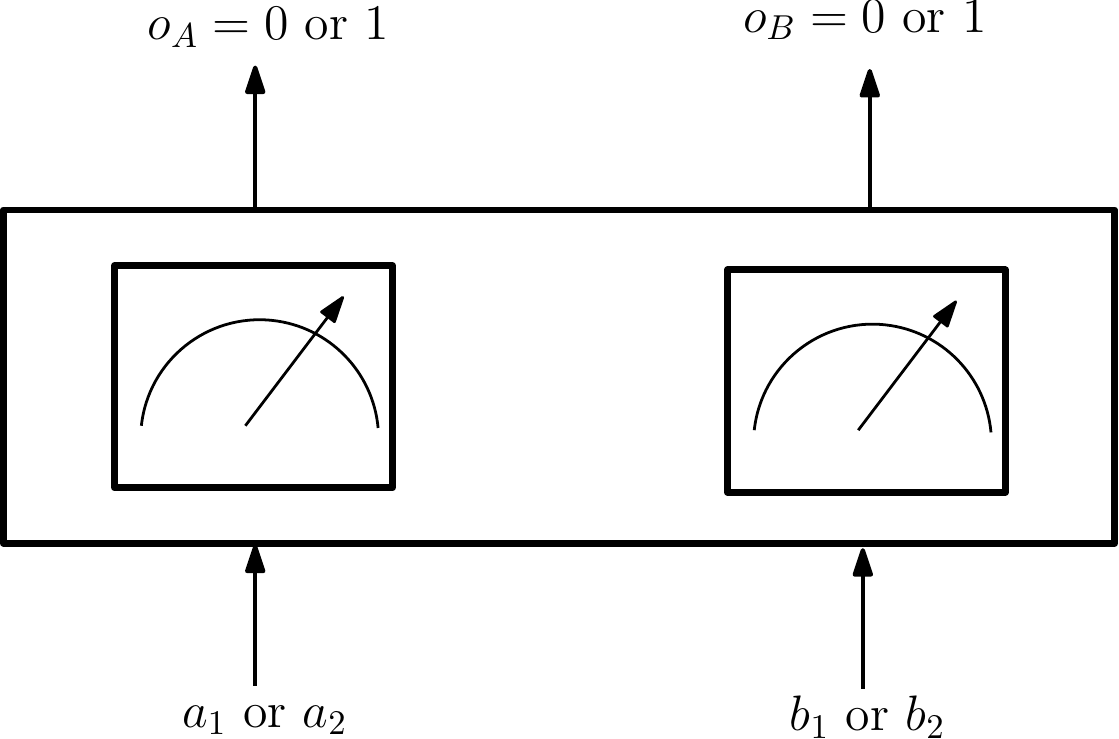}%
}{%
  \caption{A schematic representation of the experiment. 
  }\label{fig: fig1}%
}
\quad
\capbtabbox{%
\begin{tabular}{c c | c c c c}
\hline
$A$ & $B$ & $(0,0)$ & $(1,0)$ & $(0,1)$ & $(1,1)$\\
\hline
$a_1$ & $b_1$ & $\nicefrac{1}{2}$ & $0$ & $0$ & $\nicefrac{1}{2}$\\[2pt]
$a_1$ & $b_2$ & $\nicefrac{3}{8}$ &  $\nicefrac{1}{8}$ &  $\nicefrac{1}{8}$ &  $\nicefrac{3}{8}$\\[2pt]
$a_2$ & $b_1$ &  $\nicefrac{3}{8}$ &  $\nicefrac{1}{8}$ &  $\nicefrac{1}{8}$ &  $\nicefrac{3}{8}$\\[2pt]
$a_2$ & $b_2$ &  $\nicefrac{1}{8}$ &  $\nicefrac{3}{8}$ &  $\nicefrac{3}{8}$ &  $\nicefrac{1}{8}$\\[2pt]
\end{tabular}
}{%
  \caption{Bell's empirical model}\label{tab: Bell}%
}
\end{floatrow}
\end{figure}
Suppose there are two agents, Alice and Bob, who can each independently select one of two different measurements to perform on their respective share of a given system: $a_1, a_2$ for Alice, $b_1, b_2$ for Bob. 
After a measurement is performed, one of two possible outcomes, $0$ or $1$, is observed.
Note that only one measurement per party can be chosen at a time, so that there are four possible sets of jointly performable measurements, or \emph{contexts}: $\{a_1,b_1\}$, $\{a_1,b_2\}$, $\{a_2,b_1\}$, $\{a_2,b_2\}$. When a context is chosen, and measurement outcomes $o_A, o_B\in \{0,1\}$ are observed, the corresponding \emph{event} can be modelled as a tuple $\mathbf{e}=\tuple{o_A, o_B}\in O\times O$, where $O\defeq\{0,1\}$ is the \emph{set of outcomes}. 
Repeated runs of the experiment allow relative frequencies to be tabulated, which can be described as a family of probability distributions $e_C: \E(C)\rightarrow \RR_{\ge 0}$ over the events $\E(C)\defeq\prod_{m\in C}O$ of each context $C=\{a_i,b_j\}$. We shall call such a family an \emph{empirical model}. A key example of an empirical model is Bell's model \cite{Bell-thm, CHSH,Bell:Speakable}, which is displayed in Table \ref{tab: Bell}.
The probability distributions observed in Bell's model are the ones predicted by quantum theory in a particular experiment \cite{CHSH}. 
As such, they verify the \emph{no-signalling} or \emph{no-disturbance condition}, a property of all distributions arising from quantum experiments, which states that the statistics observed by one party do not depend on the other party's choice of measurement. In essence, this means that the marginals of the distributions at each context agree at their intersection.

The key feature of Bell's model is that it is \emph{contextual} or, more precisely, \emph{non-local}. While non-locality was originally formulated in terms of the non-existence of a local hidden-variable theory, it was observed originally by Fine \cite{Fine82}, and subsequently shown in a much more general form in \cite{AbramskyBrandenburger}, that this is equivalently formulated as the  impossibility of finding a global probability distribution $d: \E(X)\rightarrow \RR_{\ge 0}$ such that each empirically observed distribution $e_C$ corresponds to the marginal of $d$ with respect to the events in $C$. Hence, the properties of the system we are measuring are not predetermined, but rather depend on the measurements we choose to perform. This phenomenon cannot be witnessed in classical physics, and constitutes one of the most peculiar features of quantum theory.

By abstracting from this particular example, one can formulate an extremely general definition of non-locality and contextuality using sheaf theory. We begin by defining a general \textbf{measurement scenario} as a triple $\tuple{X, \M, (O_m)_{m\in X}}$, where $X$ is a finite set of measurement labels, $\M\subseteq \P(X)$ is a \textbf{measurement cover}, \ie the collection of \textbf{contexts}, 
and $O_m$ is the \textbf{set of outcomes} for measurement $m\in X$. 
For instance, the scenario above has $X=\{a_1,b_1,a_2,b_2\}$, $\M=\{\{a_i,b_j\}: i,j=1,2\}$ and $O_m=O=\{0,1\}$ for all $m\in X$. This is the canonical example of a so-called \textbf{Bell-type scenario}, which is a class of scenarios involving several experimenters, who can each choose to perform one measurement at a time. In a Bell-type scenario with $n$ parties, 
$X$ can be partitioned into $n$ subsets $X_1,\dots, X_n$, where $X_i$ represents the measurements available to the $i$-th party, and the contexts are of the form $\M=\{\{m_1,\dots, m_n\}: m_i\in X_i,~\forall 1\leq i\leq n\}$. 

The events of a scenario are defined by the \textbf{sheaf of events} $\E:\P(X)^{op}\rightarrow \Set$, where we have $\E(U) \defeq \prod_{m\in U}O_m$ for all $U\subseteq U'$, and the restriction maps are given by cartesian projection: for all $U\subseteq U'\subseteq X$,
\[
\E(U\subseteq U') \defeq \rho_U^{U'}: \E(U')\longrightarrow\E(U):: \mathbf{x}\longmapsto \mathbf{x}|_{U}\defeq\projn{\mathbf{x}}{U}.
\] 
It is easy to verify that $\E$ is indeed a sheaf.

Now that we have abstracted the concept of measurement scenario, we want to formally capture empirical models. In order to describe probability distributions in the most general sense, we will consider distributions over any semiring. Given a semiring $R$, we define the \textbf{$R$-distribution functor} $\D_R:\Set\rightarrow\Set$ as follows: for any set $A$ and any function $f:A\rightarrow B$, 
\[
\begin{split}
\D_R(S) &\defeq\left\{d:S\rightarrow R: |\supp(d)|<\infty \text{ and } \sum_{s\in S}d(s)=1 \right\}\\
\D_R(f) &:\D_R(A)\longrightarrow\D_R(B):: d\longmapsto \lambda b.\sum_{\substack{a\in A: \\ f(a)=b}}d(a).
 \end{split}
\]
With this premise, distributions over events of $\scenario$ can be modelled by the presheaf $\D_R\E\defeq \D_R\circ \E:\P(X)^{op}\rightarrow \Set$. Therefore, an \textbf{empirical model} will be comprised of a family $\{e_C\in\D_R(C)\}_{C\in\M}$ of local sections of $\D_R\E$. In order to capture the property of no-signalling, we require the family $\{e_C\}_{C\in\M}$ to be \textbf{compatible},\footnote{The concept of \emph{compatibility} introduced here shall not be confused with the notion of \emph{compatible measurements} of quantum theory, which describes jointly-performable measurements.} \ie such that $e_C|_{C\cap C'}=e_{C'}|_{C\cap C'}$ for all $C,C'\in\M$. In other words, for any event $\s\in\E(C\cap C')$, we have:
\[
e_C|_{C\cap C'}(\s)=\sum_{\substack{\x\in \E(C): \\ \x|_{C\cap C'}=\s}}e_C(\x)=\sum_{\substack{\y\in \E(C'): \\ \y|_{C\cap C'}=\s}}e_C(\y)=e_{C'}|_{C\cap C'}(\s),
\] 
which corresponds precisely to the property of no-signalling.

Different semirings describe different kinds of empirical models. When $R=\RR_{\ge 0}$ we say that the model is \textbf{probabilistic}; when $R=\BB$, it is \textbf{possibilistic}, which means that we are only concerned with whether events are possible or not, regardless of their probability of occurring. Notice that, by disregarding the value of individual probabilities, every probabilistic empirical model $\{e_C\in\D_{\RR_{\ge0}}\E(C)\}_{C\in\M}$ gives rise to a possibilistic model $\{\tilde{e}_C\in\D_\BB\E(C)\}_{C\in\M}$, where $\tilde{e}_C\defeq\chi_{\supp(e_C)}:\E(C)\rightarrow \{0,1\}$ is the characteristic function of the support of the distribution $e_C$. 

We say that an empirical model $\{e_C\}_{C\in \M}$ is \textbf{contextual} if there is no global section for the family $\{e_C\}_{C\in \M}$, that is, if there is no global distribution $d:\E(X)\rightarrow R$ such that $d|_C=e_C$ for all $C\in\M$. If the scenario is Bell-type, then we say that the model is \textbf{non-local}.\footnote{Note that the term \emph{local} reviewed here has a completely different meaning to the one introduced earlier in the notion of \emph{local agreement}}  Hence, non-locality is simply a special case of contextuality.

If an empirical model is possibilistic, we say that it is \textbf{logically contextual}. By extension, we say that a probabilistic model $\{e_C\}_{C\in\M}$ is logically contextual if its possibilistic collapse $\{\tilde{e}_C\}_{C\in\M}$ is logically contextual. It can be shown that this property is strictly stronger than regular (probabilistic) contextuality.

An even more restrictive kind of contextuality arises in some models. Suppose we have a global section for a possibilistic model $\{e_C\}_{C\in\M}$, \ie a distribution $d:\E(X)\rightarrow\BB$ such that $d|_C=e_C$ for all $C\in \M$. The support of $d$ is a set of global events $\g\in\E(X)$ which are compatible with the model, \ie such that $e_C(\projn{\g}{C})=1$ for all $C\in\M$. In extreme situations, not only it might not be possible to find such a global distribution $d$, but there might not even be any global event $\g$ compatible with the model at all. This property is called \textbf{strong contextuality}, and it is strictly stronger than logical contextuality. This leads to the following hierarchy of different strengths of contextuality:
\begin{equation}\label{equ: hierarchy}
\text{(probabilistic) contextuality}~~<~~\text{logical contextuality}~~<~~\text{strong contextuality}.
\end{equation}
Typical quantum examples of each of these levels are Bell's model (probabilistic), Hardy's model \cite{Hardy92:nonlocality1, Hardy93:nonlocality2} (logical) and the Greenberger-Horne-Zeilinger (GHZ) model \cite{GHZ, GHSZ90} (strong).

\subsection{Valuation algebras and sheaf theory}

Remarkably, many of the properties of valuation algebras can also be effectively captured by sheaf theory. Just like presheaves deal with the restriction and localisation of topological structures and their extendability through a `gluing' process, valuation algebras model the focus of knowledge and information, and represent the natural framework to study how local information can be extended through a `combination' process.

Let $\Phi$ be a valuation algebra on a set of variables $V$. We define a functor 
\begin{equation}\label{equ: presheaf}
\Phi: \P(V)^{op}\longrightarrow \Set,
\end{equation}
where $\Phi(S):=\Phi_S$ for all $S\subseteq V$, and restriction maps are given by the projector operator: for all $S\subseteq T\subseteq V$, we have
\[
\Phi(S\subseteq T)=\rho_S^T: \Phi_T\longrightarrow \Phi_S:: \phi\longmapsto \proj{\phi}{T}.
\]
We can easily check that $\Phi$ is a functor. Indeed, by \ref{projection}, we have, for all $S\subseteq V$ and for all $\phi\in\Phi_S$, 
\[
\rho_S^S(\phi)=\proj{\phi}{S}=\proj{\phi}{d(\phi)}\stackrel{\ref{projection}}{=}\phi,
\]
and, by \ref{transitivity}, for all $S\subseteq T\subseteq U\subseteq V$ and $\phi\in\Phi_U$,
\[
\rho_S^T\circ\rho_T^U(\phi)=\proj{\left(\proj{\phi}{T}\right)}{S}\stackrel{\ref{transitivity}}{=}\proj{\phi}{S}=\rho_S^U(\phi).
\]
This sheaf-theoretic perspective allows to capture the restriction or localisation of the information carried by a valuation algebra. 


It is important to keep in mind that the presheaf description of a valuation algebra only accounts for the restriction process. The combination operation is unique to each valuation algebra, although there are some general constructions, as we shall see later on in Section \ref{sec: disagreement and inference problems}\ref{sec: combination}.

With this premise, it is easy to convert the definitions of local and global agreement in sheaf-theoretic terms. A locally agreeing knowledgebase $K=\{\phi_1,\dots,\phi_n\}$ is nothing but a compatible family of local sections of $\Phi$. On the other hand, $K$ agrees globally if and only if there is a global section for $K$. 

\subsection{Contextuality and disagreement}\label{sec: disagreement and contextuality}

It is now time to reveal the link between contextuality and disagreement underpinning the theory presented thus far. We have purposely kept the connection implicit until now, so as to show how the theory of disagreement can be developed completely independently of the concept of contextuality.

\subsubsection*{No-signalling and local agreement.}
Let $e=\{e_C\in\D_R\E(C)\}_{C\in\M}$ be an empirical model over a measurement scenario $\scenario$. Let us take $X$ as a set of variables on which to build a suitable valuation algebra. for each measurement $m\in X$, we let $\Omega_m\defeq O_m$, so that $\Omega_C=\E(C)$ for any context $C\in\M$. 


We can interpret each $e_C$ as a valuation of the algebra $\Phi$ of $R$-potentials. From this viewpoint, the empirical model $e$ is nothing but a knowledgebase of $\Phi$. Then, the property of no-signalling -- \ie the compatibility of $\{e_C\}_{C\in\M}$ -- corresponds precisely to local agreement. Hence, to summarise, one can think of an empirical model as a locally agreeing knowledgebase of $\Phi$.

\subsubsection*{Contextuality and global disagreement.} 

By considering no-signalling empirical models as locally agreeing knowledgebases of the algebra of $R$-potentials, a striking connection with the theory of disagreement arises: non-locality and contextuality are just special instances of a locally agreeing knowledgebase which disagrees globally:

\begin{theorem}[See \ref{proof truth} for the proof]\label{thm: main connection}
Let $e=\{e_C\}_{C\in\M}$ be an empirical model. Then $e$ is contextual if and only if the locally-agreeing knowledgebase $K=\{e_C\}_{C\in\M}\subseteq \Phi$ disagrees globally.
\end{theorem}
%
Therefore, from a structural perspective, these counterintuitive phenomena of quantum physics are exhibiting exactly the same kind of behaviour as  the examples introduced in Section \ref{sec: disagreement}, and, indeed, as any instance of local agreement vs global disagreement arising from the valuation algebra framework. This result is a major generalisation of the connections observed in \cite{Abramsky12:databases}, \cite{AbramskyGottlobKolaitis2013:RobustConstraintSatisfactionAndLHV}, and \cite{AbramskyEtAl:CCP2015} which are limited to relational databases, CSPs, and logical paradoxes respectively. It further proves that contextuality is not a phenomenon limited to quantum physics, but it is a general concept which pervades various domains, most of which are completely unrelated to quantum theory. 

This connection can be further explored to include all the levels of strength of contextuality presented in \eqref{equ: hierarchy}, as we shall see in Section \ref{sec: disagreement and logical contextuality}. Moreover, it paves the way for the application of methods of generic inference to the study of contextuality via the reformulation of the problem of detecting contextuality as an inference problem, which is also presented in Section \ref{sec: disagreement and logical contextuality}.


%
%
%
%
%
%
%
%
%


\section{Disagreement, complete disagreement and inference problems}\label{sec: disagreement and inference problems}

Studying global disagreement amounts to looking for a global truth, which is shared by all the sources of information. It is thus natural to ask whether it is possible to recover it from the collective information of the sources and the structure of the valuation algebra. It turns out that, in a variety of situations, the global truth valuation can appear only in one form, which makes the problem of finding it significantly easier and, crucially, equivalent to an inference problem. In order to prove this, we will need to introduce the concept of an \emph{ordered valuation algebra}.


\subsection{Ordered valuation algebras}\label{sec: ordered valuation algebras}
Given a valuation algebra $\Phi$ on a set of variables $V$, and two valuations $\phi,\psi\in\Phi_S$ for some $S\subseteq V$, one could raise the following question: \emph{how does the information carried by $\phi$ compare to the one carried by $\psi$?} \emph{Is there a way of quantifying the amount of information represented by a valuation?} The answer to this question is given by extending the present framework to the one of \emph{ordered valuation algebras} \cite{haenni2004ordered}. An ordered valuation algebra is a valuation algebra equipped with a completeness relation $\preceq$, which aims to capture how informative a valuation is with respect to others.

\begin{defn}
Let $\Phi$ be a valuation algebra with null elements on a set of variables $V$. Then, $\Phi$ is an \textbf{ordered valuation algebra}\index{valuation algebra!ordered valuation algebra|textbf} if there exists a partial order $\preceq$ on $\Phi$ such that the following additional axioms are verified:

\begin{itemize}
\item[(A10)] \emph{Partial order}: For all $\phi,\psi\in\Phi$, $\phi\preceq\psi$ implies $d(\phi)=d(\psi)$. Moreover, for every $S\subseteq V$ and $\Psi\subseteq\Phi_S$, the infimum $\inf(\Psi)$ exists. 
\item[(A11)] \emph{Null element}: For all $S\subseteq V$, we have $\inf(\Phi_S)=z_S.$
\item[(A12)] \emph{Monotonicity of combination}: For all $\phi_1,\phi_2,\psi_1,\psi_2\in\Phi$ such that $\phi_1\preceq\phi_2$ and $\psi_1\preceq\psi_2$ we have $\phi_1\otimes\psi_1\preceq\phi_2\otimes\psi_2.$
\item[(A13)] \emph{Monotonicity of projection}: For all $\phi,\psi\in\Phi$, if $\phi\preceq\psi$ then $\proj{\phi}{S}\preceq\proj{\psi}{S},$
for all $S\subseteq d(\phi)=d(\psi)$. 
\end{itemize}
\end{defn}
It can be shown that all the instances of valuation algebras presented in the previous sections can be ordered. For instance, the algebra of relational databases has an order structure simply given by inclusions.
%
%
We can incorporate some of the axioms in the structure of the presheaf \eqref{equ: presheaf} by simply rewriting it as $
\Phi:\P(V)^{op}\longrightarrow \Pos$,
where $\Pos$ denotes the category of posets and monotone maps. 

\subsection{A general construction for composition}\label{sec: combination}

An interesting aspect brought to light by the order structure of valuation algebras is that the composition laws of many algebras are uniquely characterised by the same categorical construction. 

Let $\Phi$ be an ordered valuation prealgebra, viewed as a presheaf. 
Thanks to the universal property of products of the category $\Set$, we have, for all $S,T\subseteq V$, the following diagram:
\begin{center}
\begin{tikzpicture}[auto]
\matrix (m) [matrix of math nodes, row sep=6em, column sep=6em, text height=1.5ex, text depth=0.25ex]
{
\Phi(S) & \Phi(S)\times\Phi(T) & \Phi(T) \\
 & \Phi(S\cup T) & \\
};

\path[-stealth, font=\scriptsize]
(m-1-2) edge node[swap] {$\pi_1$} (m-1-1)
edge node[auto] {$\pi_2$} (m-1-3)
(m-2-2) edge node[auto] {$\rho_S^{S\cup T}$} (m-1-1)
edge node[swap] {$\rho_T^{S\cup T}$} (m-1-3);

\path[->,dashed, font=\scriptsize]
(m-2-2) edge node[auto] {$\substack{\tuple{\rho_S^{S\cup T},\rho_T^{S\cup T}}\\ ~~ \\ ~~ \\~~ \\~~ \\~~ \\~~ \\~~ }$} (m-1-2);
\end{tikzpicture}
\end{center}
This leads to the following definition:
an \textbf{adjoint valuation algebra}\index{valuation algebra!adjoint valuation algebra|textbf} is an ordered valuation algebra $\Phi$ such that its combination operation $\otimes$ is the right adjoint of the map $\tuple{\rho_S^{S\cup T},\rho_T^{S\cup T}}$, defined in the diagram above. In this case, $\otimes$ is the unique map such that
\begin{equation}\label{equ: adjoint equ}
\mathsf{id}_{\Phi(S\cup T)}\leq \otimes\circ\tuple{\rho_S^{S\cup T},\rho_T^{S\cup T}}, ~~~~~~~ \tuple{\rho_S^{S\cup T},\rho_T^{S\cup T}}\circ\otimes\leq\mathsf{id}_{\Phi(S)\times \Phi(T)},
\end{equation}
where $\leq$ is the pointwise order inherited by the partial order $\preceq$ of the algebra.

Adjoint valuation algebras are extremely common: relational databases, CSPs, propositional logic, predicate logic, linear equations, linear inequalities, can all be proven to be adjoint (see Proposition \ref{prop: appendix} in \ref{proof truth}\ref{sec: adjoint info algebras}).

\subsubsection{Detecting disagreement is an inference problem}\label{sec: disagreement IP}
The most important aspect of adjoint valuation algebras is that, given a globally agreeing knowledgebase, if a truth valuation exists, then it can be obtained by simply combining all the available information. This is the content of the following proposition, which generalises Proposition 2.3 in \cite{Abramsky12:databases}. 

\begin{proposition}[See \ref{proof truth} for the proof]\label{prop: truth}
Let $\Phi$ be an adjoint valuation algebra on a set of variables $V$. Let $K=\{\phi_1,\dots,\phi_n\}\subseteq\Phi$ be a knowledgebase. Let $\gamma=\bigotimes_{i=1}^n \phi_i.$
Then $\phi_1,\dots,\phi_n$ agree globally if and only if $\proj{\gamma}{d(\phi_i)}=\phi_i$. In this case, $\gamma$ is the most informative of all the possible truth valuations. 
\end{proposition}

Thanks to this proposition, the quest for a global truth valuation becomes a much easier task. This, in turn, makes the problem of detecting disagreement significantly simpler. In fact, we can reformulate it as an inference problem. Given a knowledgebase $\{\phi_1,\dots,\phi_n\}$, it is sufficient to solve the problem
\begin{equation}\label{equ: disagreement IP}
\proj{(\phi_1\otimes\dots\otimes \phi_n)}{d(\phi_i)}
\end{equation}
for all $1\leq i\leq n$. Then, the knowledgebase agrees globally if and only if the solution to each problem is $\phi_i$. 

%
%
%
%

\subsection{Complete disagreement}\label{sec: complete disagreement}
Let $\phi_1,\phi_2\in\Phi$ be two valuations of an adjoint information algebra $\Phi$, and let $d(\phi_1)=S$, $d(\phi_2)=T$, with $S\cap T\neq\emptyset$. By Proposition \ref{prop: truth}, to say that $\phi_1$ and $\phi_2$ disagree amounts to say that not \emph{all} the information carried by $\phi_1$ and $\phi_2$ can be preserved by combining them. However, \emph{some} of this information is preserved, namely the quantities $\psi_1\defeq \proj{(\phi_1\otimes\phi_2)}{S}$ and $\psi_2\defeq \proj{(\phi_1\otimes\phi_2)}{T}$. Indeed, $\psi_1\preceq\phi_1$ and $\psi_2\preceq\phi_2$ represent exactly
the portion of information on which the original valuations \emph{do} agree. This can be easily shown by arguing that $\psi_1$ and $\psi_2$ agree on their common variables:
\[
\proj{\psi_1}{S\cap T}=\proj{\left(\proj{(\phi_1\otimes\phi_2)}{S}\right)}{S\cap T}=
\proj{(\phi_1\otimes\phi_2)}{S\cap T}=\proj{\left(\proj{(\phi_1\otimes\phi_2)}{T}\right)}{S\cap T}=\proj{\psi_2}{S\cap T}.
\]
However, there may be situations where $\psi_1$ and $\psi_2$ are null elements of the algebra. This corresponds to a situation where $\phi_1$ and $\phi_2$ disagree \emph{completely}. In this case, we have $\phi_1\otimes\phi_2=z_{S\cup T}.$
The liar cycle of Example \ref{exmp: liar} gives a compelling example of complete disagreement. Let us compute the global valuation $\gamma:=\bigotimes_{i=1}^n \phi_n$. We have
\[
\bigotimes_{i=1}^{n-1}\phi_i=\left\{\vi\in\Models_{V}: \vi(\tuple{s_1,\dots, s_n})=\tuple{0,0,\dots,0} \text{ or } \tuple{1,1,\dots,1} \right\}
\]
Hence,
\[
\begin{split}
\gamma &=\left\{\vi\in\Models_{V}: \vi(\tuple{s_1,\dots, s_n})=\tuple{0,0,\dots,0} \text{ or } \tuple{1,1,\dots,1} \right\}\\
&\otimes\left\{\vi\in\Models_{\{s_1,s_n\}}:\vi(\tuple{s_1,s_n})=\tuple{1,0}\text{ or } \tuple{0,1}\right\}=\emptyset=z_V,
\end{split}
\]
Hence, we conclude that this knowledgebase disagrees completely, despite agreeing locally.

In light of this discussion, we introduce the following definition: we say that $\phi_1,\dots,\phi_n$ \textbf{disagree completely}\index{disagreement (valuation algebras)!complete disagreement|textbf} if $\gamma:=\bigotimes_{i=1}^n\phi_i=z_V$, or, equivalently by axiom \ref{nullity}, if there exists a $1\leq i\leq n$ such that
\[
\proj{(\phi_1\otimes\dots\otimes \phi_n)}{d(\phi_i)}=z_{d(\phi_i)}.
\]

\section{Disagreement and possibilistic forms of contextuality}\label{sec: disagreement and possibilistic forms of contextuality}\label{sec: disagreement and logical contextuality}
In this section we investigate the link between disagreement and possibilistic forms of contextuality, \ie logical and strong.
This will lead us to establish a connection between strong contextuality and complete disagreement, and to prove that the problem of detecting logical forms of contextuality can be rephrased as an inference problem for the valuation algebra of relational databases. 

An alternative definition of possibilistic empirical models, developed in \cite{AbramskyEtAl:CCP2015}, will be particularly useful: a possibilistic empirical model on a scenario $\scenario$ can be equivalently defined as a subpresheaf $\mathcal{S}\subseteq\mathcal{E}$ which verifies the following properties:
\begin{enumerate}
\item\label{cond1} $\mathcal{S}(C)\neq \emptyset$ for all $C\in\mathcal{M}$
\item\label{cond2} $\mathcal{S}$ is \emph{flasque beneath the cover}, i.e. the map $\mathcal{S}(U\subseteq U')$ is surjective whenever $U\subseteq U'\subseteq C$ for some $C\in\mathcal{M}$. 
\item\label{cond3} Every compatible family $\{\s_C\in\mathcal{S}(C)\}_{C\in\mathcal{M}}$ induces a global section $\g\in\mathcal{S}(X)$ such that $\g\mid_C=\s_C$ for all $C\in\M$. Note that this global section is unique since $\mathcal{S}$ is a subpresheaf of the sheaf $\mathcal{E}$. 
\end{enumerate}
Then, possibilistic forms of contextuality can be shown to correspond to the following properties \cite{AbramskyEtAl:CCP2015}. Let $\S:\P(X)^{op}\rightarrow\Set$ be a possibilistic model on a scenario $\scenario$. 
\begin{itemize}
\item Given a context $C\in\mathcal{M}$ and a section $\s\in\mathcal{S}(C)$, $\mathcal{S}$ is \textbf{logically contextual at $\s$}, or $\mathsf{LC}(\mathcal{S}, \s)$, if $\s$ is not a member of any compatible family. 
\item $\mathcal{S}$ is \textbf{strongly contextual}, or $\mathsf{SC}(\mathcal{S})$, if $\mathsf{LC}(\mathcal{S}, \s)$ for all $\s$. In other words, by condition \ref{cond3}, $\S$ does not have any global section: $\mathcal{S}(X)=\emptyset$.
\end{itemize}

Let $\scenario$ be a scenario, and consider the valuation algebra $\Phi$ of relational databases on the set of variables $X$, where we define the frame of each variable $m\in X$ to be $\Omega_m\defeq O_m$. In particular, this means that for any context $C\in\M$ we have $\Omega_C=\E(C)$.
Now, suppose we have a possibilistic empirical model $\S:\P(X)^{op}\rightarrow\Set$ over $\scenario$. 
We associate to it the knowledgebase $K=\{\S(C)\subseteq\E(C)=\Omega_C\}_{C\in\M}\subseteq \Phi.$
\begin{proposition}[See \ref{proof truth} for the proof]\label{prop: K agrees locally}
The knowledgebase $K$ agrees locally.
\end{proposition}
This means that a possibilistic empirical model corresponds to a locally-agreeing knowledgebase of the algebra or relational databases. We now have all the elements to establish the connection between global disagreement and logical forms of contextuality


\begin{proposition}[See \ref{proof truth} for the proof]\label{prop: LCD}
The knowledgebase $K$ disagrees globally if and only if $\S$ is logically contextual. It disagrees completely if and only if $\S$ is strongly contextual.
\end{proposition}
Notice that the first part of the proposition reiterates the idea, expressed in Section \ref{sec: NL C}\ref{sec: disagreement and contextuality}, that contextuality in an empirical model corresponds to an instance local agreement vs global disagreement.\footnote{In fact the two results coincide, as one can show that the algebra of relational databases is equivalent to the one of Boolean potentials via the isomorphism $\P(X)\cong 2^X$.} The statement connecting strong contextuality with complete disagreement, on the other hand, is a new addition to the theory. 

\subsection{Detecting logical and strong contextuality is an inference problem}
Thanks to Proposition \ref{prop: LCD} and the results of the previous sections, we can easily translate the problem of detecting logical and strong contextuality into inference problems. This aspect is particularly important since it allows to apply the numerous efficient algorithms developed by the long-established theory of generic inference.  
The following proposition follows immediately from Proposition \ref{prop: LCD} and the results of Sections \ref{sec: disagreement and inference problems}\ref{sec: combination}\ref{sec: disagreement IP} and \ref{sec: disagreement and inference problems}\ref{sec: complete disagreement}:

\begin{proposition}\label{prop: IPC}
Let $\S:\P(X)\rightarrow\Set$ be an empirical model over a scenario $\scenario$. Then
\begin{itemize}
\item The model $\S$ is logically contextual if and only if there exists a context $\overline{C}\in\M$ such that the inference problem
\begin{equation}\label{equ: IPLC}
\proj{\left(\bigotimes_{C\in\M}\S(C)\right)}{\overline{C}}
\end{equation}
for the algebra of relational databases has a solution different from $\S(\overline{C})$. 
\item The model $\S$ is strongly contextual if and only if, there exists a $\overline{C}\in\M$, such that the solution to the inference problem \eqref{equ: IPLC} 
is $z_{\overline{C}}$. Notice that, by \eqref{equ: z} in \ref{nullity}, this is equivalent to saying that $\bigotimes_C\S(C)=z_X$.
\end{itemize}
\end{proposition}

It follows that, in order to determine whether a model $\S$ is strongly contextual, one has to solve a single inference problem \eqref{equ: IPLC}: if the result is $z_{\overline{C}}$, then $\S$ is strongly contextual. On the other hand, to determine whether $\S$ is logically contextual, one has to solve $|\M|$ distinct problems, in the worst case scenario. 
This reformulation is particularly important since it allows to use efficient algorithms of generic inference to solve \eqref{equ: IPLC}. We have already started to explore this possibility in separate work, and it has been possible to develop new algorithms for the detection of logical forms of contextuality which significantly outperform the current state of the art \cite{CaruPhD2019, AbramskyCaru2019Upcoming}.

\section{Conclusions}\label{sec: conclusion}
We have presented a general definition of different forms of disagreement between information sources in the abstract framework of valuation algebras. In particular, we identified three kinds of disagreement: local, global and complete, and presented examples of each of them using different valuation algebras. A particular attention has been given to instances of knowledgebases which agree locally but disagree globally. By recovering the valuation algebraic formalism in sheaf-theoretic terms, we showed that sheaf-theoretic contextuality is simply a special case of such a knowledgebase, where the valuation algebra in question is the one of $R$-potentials, while strong contextuality is a special case of complete disagreement for the algebra of relational databases. This result is a vast generalisation of the previously observed connections between contextuality and relational databases, constraint satisfaction problems, and logical paradoxes, and constitutes a promising attempt to establish a general theory of contextual semantics. The main advantage of such an abstract and flexible treatment is that it significantly widens the scope for the observation of contextual behaviour, and could potentially lead to the transfer of results and methods for disagreement across the many different fields captured by the valuation algebraic framework. 

In separate work, we have started to explore this potential by applying popular inference algorithms such as the fusion \cite{shenoy1992valuation} and collect \cite{SHENOY:1990aa} methods to detect contextuality in empirical models. This problem is notoriously complex from a computational perspective \cite{AbramskyGottlobKolaitis2013:RobustConstraintSatisfactionAndLHV}, and very few algorithms have been developed for this specific task. Yet computational explorations of logical forms of contextuality, however limited, have proved very useful for the advancement of the theory \cite{AbramskyGottlobKolaitis2013:RobustConstraintSatisfactionAndLHV, MansfieldFritz11:Hardy}, and would greatly benefit from any algorithmic improvement. For this reason, the theory introduced in this paper represents a major opportunity to enhance our understanding of contextuality.

\competing{The authors declare no competing interests.}
\aucontribute{Both the authors contributed equally to this work. Their names are listed in alphabetical order.}
\ack{The authors would like to thank Rui Soares Barbosa for helpful discussions.}
\funding{Support from the following is gratefully acknowledged: EPSRC EP/N018745/1, `Contextuality as a Resource in Quantum Computation' (SA); EPSRC Doctoral Training Partnership, and  Oxford--Google Deepmind Graduate Scholarship (GC).}

\appendix

\section{Proofs}\label{proof truth}

\subsection{Theorem \ref{thm: main connection}}
\begin{proof}
If $e$ is non-contextual, then there exists a global distribution $d:\E(X)\rightarrow R$ such that $d|_C=e_C$ for all $C\in\M$. Hence $d$ is a global truth valuation for the knowledgebase $K$. Conversely, suppose $K$ agrees globally, which means that there exists a global $R$-potential $d:\Omega_X=\E(X)\rightarrow R$ such that $\proj{d}{C}=e_C$ for all $C\in \M$. The only thing we need to prove is that $d$ is an $R$-distribution, \ie that it is normalised. Let $\ast$ denote the unique $\emptyset$-tuple and let $C\in\M$ be an arbitrary context. We have:
\[
\begin{split}
\sum_{\g\in\E(X)}d (\g) &=\sum_{\substack{\g\in\E(X): \\ \projn{\g}{\emptyset}=\ast}}d (\g)=\proj{d}{\emptyset}(\ast)=\proj{\left(\proj{d}{C}\right)}{\emptyset}(\ast)=\proj{(e_C)}{\emptyset}(\ast)\\
&= \sum_{\substack{\x\in\E(C): \\ \projn{\x}{\emptyset} = \ast}}e_{C}(\x)= \sum_{\x\in\E(C)}e_{C}(\x)=1.
\end{split}
\]
\end{proof}

\subsection{Adjoint information algebras}\label{sec: adjoint info algebras}

\begin{proposition}\label{prop: appendix}
The algebra of information sets related to a family $\{\tuple{\L_Q,\Models_Q, \models_Q}\}_{Q\subseteq V}$ defined over any language $\L$ and set of models $\Models$ is adjoint. 
\end{proposition}

\begin{proof}
It is sufficient to prove \eqref{equ: adjoint equ}:
\begin{itemize}
\item Let $M\subseteq \Models_{Q\cup U}$, where $Q,U\subseteq V$.
\[
\begin{split}
\proj{M}{Q}\otimes\proj{M}{U} &=\{\projn{\x}{Q}:\x\in M\}\otimes\{\projn{\x}{U}:\x\in M\}\\
&=\left\{\vi\in\Models_{Q\cup U}:\exists\x,\y\in M\text{ s.t. }(\projn{\vi}{Q}=\projn{\x}{Q})\wedge(\projn{\vi}{U}=\projn{\y}{U})\right\}.
\end{split}
\]
Then, clearly, $M\subseteq \proj{M}{Q}\otimes\proj{M}{U}$.
\item Now, let $M_1\subseteq \Models_Q$ and $M_2\subseteq \Models_U$. We have
\[
\begin{split}
\proj{(M_1\otimes M_2)}{Q} &=\proj{\{\vi\in\Models_{Q\cup U}: (\projn{\vi}{Q}\in M_1)\wedge (\projn{\vi}{U}\in M_2) \}}{Q}\\
&=\left\{\projn{\vi}{Q}: \vi\in\Models_{Q\cup U}\wedge (\projn{\vi}{Q}\in M_1)\wedge (\projn{\vi}{U}\in M_2)\right\}\subseteq M_1.
\end{split}
\]
One proves that $\proj{(M_1\otimes M_2)}{U}\subseteq M_2$ in the same way.
\end{itemize}
\end{proof}

\subsection{Proposition \ref{prop: truth}}
\begin{proof}[Proof of Proposition \ref{prop: truth}]
Suppose $\delta\in\Phi_V$ is a truth valuation for $K$. Then we have
\[
\delta\preceq\bigotimes_{i=1}^n\proj{\delta}{d(\phi_i)}=\bigotimes_{i=1}^n\phi_i=\gamma
\]
Moreover, because projection is monotone by axiom (A13), we have
\[
\phi_i\preceq\proj{\delta}{d(\phi_i)}\stackrel{\text{(A13)}}{\preceq} \proj{\gamma}{d(\phi_i)}\preceq\phi_i.
\]
Thus $\gamma$ is a truth valuation for $K$.
\end{proof}
\subsection{Proposition \ref{prop: K agrees locally}}\label{prop: local}

\begin{proof}[Proof of Proposition \ref{prop: K agrees locally}]
Let $\x\in\E(C\cap C')=\Omega_{C\cap C'}$. We have
\[
\begin{split}
\proj{i_{\S(C)}}{C\cap C'}(\x) &=\max_{\y\in \E(C\setminus C')}i_{\S(C)}(\x,\y)=
\begin{cases}
1 & \text{ if } \exists\y\in\E(C\setminus C'): \tuple{\x,\y}\in\S(C),\\
0 & \text{ otherwise.}
\end{cases}\\
&=\begin{cases}
1 & \text{ if } \x\in\S(C\cap C')\\
0 & \text{otherwise}
\end{cases}=i_{\S(C\cap C')}
\end{split}
\]
where the penultimate equality follows from the fact that $\S$ is flasque beneath the cover, and $C\cap C'\subseteq C\in\M$. 
With the same argument we show that $\proj{i_{\S(C')}}{C\cap C'}=i_{\S(C\cap C')}$, and we conclude
\[
\proj{i_{\S(C)}}{C\cap C'}=i_{\S(C\cap C')}=\proj{i_{\S(C')}}{C\cap C'},
\]
which means that $K$ agrees locally.
\end{proof}

\subsection{Proposition \ref{prop: LCD}}

\begin{proof}[Proof of Proposition \ref{prop: LCD}]\
We will denote $\Gamma:=\bigotimes_{C\in\M}S(C)$.

\begin{itemize}
\item Suppose $K$ disagrees globally. By Proposition \ref{prop: truth}, there exists a context $C_0\in\M$ such that $\proj{\Gamma}{C_0}\neq \S(C_0)$. Since the algebra of relational databases is adjoint, by Proposition \ref{prop: truth}, this implies $\proj{\gamma}{C_0}\subset \S(C_0)$, which means that there exists a local section $\x\in\S(C_0)$ such that $\x\notin\proj{\Gamma}{C_0}$. We will now show that $\S$ is logically contextual at $\x$. Suppose $\neg\LC(\S,\x)$ by contradiction. Then there exists a global section $\g\in\S(X)$ such that $\projn{\g}{C_0}=\x$. Because $\g\in\S(X)$, we have $\projn{\g}{C}\in\S(C)$ for all $C\in\M$, which implies $\g\in\Gamma$. Then, $\x=\projn{g}{C_0}\in\proj{\Gamma}{C_0}$, which is a contradiction.

Now, Suppose $\S$ is logically contextual at a section $\x\in\S(C_0)$. If $\x\in\proj{\Gamma}{C_0}$, then there exists a $\g\in\Gamma$ such that $\projn{\g}{C_0}=\x$. Since $\g\in\Gamma$, we have $\projn{\g}{C}\in\S(C)$ for all $C\in\M$. Thus, by condition \eqref{cond3} of the definition of a possibilistic empirical model, $\g\in\S(X)$, which means that $\g$ is a global is a global section extending $\x$. This contradicts the fact that $\S$ is logicallly contextual at $\x$. We conclude that $\x\notin\proj{\Gamma}{C_0}$. Thus $\proj{\Gamma}(C_0)\neq\S(C_0)$, hence $K$ disagrees globally. 
\item We will now prove that $K$ disagrees completely if and only if $\S$ is strongly contextual. Recall that the null element of the algebra is the emptyset $z_X=\emptyset$. We have
\[
\begin{split}
\neg\SC(\S) &\Leftrightarrow \exists \g\in\S(X)\Leftrightarrow \exists \g\in\E(X): \projn{\g}{C}\in\S(C)~\forall C\in\M\\
&\Leftrightarrow \exists\g\in\S(X):\g\in\Gamma\Leftrightarrow\Gamma\neq 0.
\end{split}
\]
\end{itemize}
\end{proof}


\begin{thebibliography}{10}

\bibitem{HowardEtAl:ContextualityMagic}
Howard M, Wallman J, Veitch V, Emerson J.
\newblock Contextuality supplies the {`magic'} for quantum computation.
\newblock Nature. 2014 06;510(7505):351--355.

\bibitem{Raussendorf:ContextualityInMBQC}
Raussendorf R.
\newblock Contextuality in measurement-based quantum computation.
\newblock Physical Review A. 2013 Aug;88(2):022322.

\bibitem{AbramskyBrandenburger}
Abramsky S, Brandenburger A.
\newblock The sheaf-theoretic structure of non-locality and contextuality.
\newblock New Journal of Physics. 2011;13(11):113036.
\newblock Eprint available at
  {\href{http://arxiv.org/abs/1102.0264}{arXiv:1102.0264 [quant-ph]}}.

\bibitem{Abramsky12:databases}
Abramsky S.
\newblock Relational databases and {B}ell's theorem.
\newblock In: Tannen V, Wong L, Libkin L, Fan W, Tan WC, Fourman M, editors. In
  search of elegance in the theory and practice of computation: Essays
  dedicated to {P}eter {B}uneman. vol. 8000 of Lecture Notes in Computer
  Science. Springer Berlin Heidelberg; 2013. p. 13--35.
\newblock Eprint available at
  {\href{http://arxiv.org/abs/1208.6416}{arXiv:1208.6416 [cs.LO]}}.

\bibitem{AbramskyGottlobKolaitis2013:RobustConstraintSatisfactionAndLHV}
Abramsky S, Gottlob G, Kolaitis PG.
\newblock Robust constraint satisfaction and local hidden variables in quantum
  mechanics.
\newblock In: Rossi F, editor. Proceedings of the Twenty-Third International
  Joint Conference on Artificial Intelligence. AAAI Press; 2013. p. 440--446.

\bibitem{AbramskySadrzadeh2014:SemanticUnification}
Abramsky S, Sadrzadeh M.
\newblock Semantic unification: A sheaf theoretic approach to natural language.
\newblock In: Casadio C, Coecke B, Moortgat M, Scott P, editors. Categories and
  types in logic, language, and physics: Essays Dedicated to {J}im {L}ambek on
  the occasion of his 90th birthday. vol. 8222 of Lecture Notes in Computer
  Science. Springer; 2014. p. 1--13.
\newblock Eprint available at
  {\href{http://arxiv.org/abs/1403.3351}{arXiv:1403.3351 [cs.CL]}}.

\bibitem{AbramskyEtAl:CCP2015}
Abramsky S, Barbosa RS, Kishida K, Lal R, Mansfield S.
\newblock Contextuality, cohomology and paradox.
\newblock In: Kreutzer S, editor. 24th EACSL Annual Conference on Computer
  Science Logic (CSL 2015). vol.~41 of Leibniz International Proceedings in
  Informatics (LIPIcs). Dagstuhl, Germany: Schloss Dagstuhl--Leibniz-Zentrum
  fuer Informatik; 2015. p. 211--228.

\bibitem{AbramskyBarbosaCaruPerdrix:AvNTriple}
Abramsky S, Barbosa RS, Car{\`u} G, Perdrix S.
\newblock A complete characterization of all-versus-nothing arguments for
  stabilizer states.
\newblock Philosophical Transactions of the Royal Society of London A:
  Mathematical, Physical and Engineering Sciences. 2017;375(2106).
\newblock Available from:
  \url{http://rsta.royalsocietypublishing.org/content/375/2106/20160385}.

\bibitem{DeSilva2017}
de~Silva N. Logical paradoxes in quantum computation; 2017.
\newblock Eprint available at
  {\href{http://arxiv.org/abs/1709.00013}{arXiv:1709.00013 [quant-ph]}}.

\bibitem{Shenoy:1989aa}
Shenoy PP.
\newblock A valuation-based language for expert systems.
\newblock International Journal of Approximate Reasoning. 1989;3(5):383--411.
\newblock Available from:
  \url{http://www.sciencedirect.com/science/article/pii/0888613X89900091}.

\bibitem{SHENOY:1990aa}
Shenoy PP, Shafer G, Shachter RD, Levitt TS, Kanal LN, Lemmer JF.
\newblock In: Axioms for Probability and Belief-Function Propagation. vol.~9.
  North-Holland; 1990. p. 169--198.
\newblock Available from:
  \url{http://www.sciencedirect.com/science/article/pii/B9780444886507500196}.

\bibitem{kohlas1996information}
Kohlas J, St{\"a}rk RF.
\newblock Information algebras and information systems.
\newblock Citeseer; 1996.

\bibitem{kohlas2000computation}
Kohlas J, Shenoy PP.
\newblock Computation in valuation algebras.
\newblock In: Handbook of defeasible reasoning and uncertainty management
  systems. Springer; 2000. p. 5--39.

\bibitem{shenoy1994consistency}
Shenoy PP.
\newblock Consistency in valuation-based systems.
\newblock ORSA Journal on Computing. 1994;6(3):281--291.

\bibitem{KohlasEtAl1999}
Kohlas J, Haenni R, Moral S.
\newblock Propositional information systems.
\newblock Journal of Logic and Computation. 1999;9(5):651--681.
\newblock Available from: \url{http://dx.doi.org/10.1093/logcom/9.5.651}.

\bibitem{shafer1991local}
Shafer GR, Shenoy PP. Local computation in hypertrees; 1991.
\newblock Tech. Rept. School of Business, University of Kansas.

\bibitem{kohlas2003information}
Kohlas J.
\newblock Information Algebras.
\newblock Springer; 2003.

\bibitem{pouly2008generic}
Pouly M.
\newblock A generic framework for local computation.
\newblock Universit{\'e} de Fribourg; 2008.

\bibitem{PoulySoftware2010}
Pouly M.
\newblock NENOK --- A software architecture for generic inference.
\newblock International Journal on Artificial Intelligence Tools.
  2010;19(01):65--99.
\newblock Available from: \url{https://doi.org/10.1142/S0218213010000042}.

\bibitem{pouly2012generic}
Pouly M, Kohlas J.
\newblock Generic Inference: A Unifying Theory for Automated Reasoning.
\newblock John Wiley \& Sons; 2012.

\bibitem{spekkens2005contextuality}
Spekkens RW.
\newblock Contextuality for preparations, transformations, and unsharp
  measurements.
\newblock Physical Review A. 2005;71(5):052108.

\bibitem{csw2014graphtheoretic}
Cabello A, Severini S, Winter A.
\newblock Graph-theoretic approach to quantum correlations.
\newblock Physical Review Letters. 2014;112:040401.

\bibitem{acin2015combinatorial}
Ac{\'\i}n A, Fritz T, Leverrier A, Sainz AB.
\newblock A combinatorial approach to nonlocality and contextuality.
\newblock Communications in Mathematical Physics. 2015;334(2):533--628.

\bibitem{dzhafarov2016context}
Dzhafarov EN, Kujala JV.
\newblock Context--content systems of random variables: The
  Contextuality-by-Default theory.
\newblock Journal of Mathematical Psychology. 2016;74:11--33.

\bibitem{ullman1984principles}
Ullman JD.
\newblock Principles of database systems.
\newblock Galgotia publications; 1984.

\bibitem{wilson1999logical}
Wilson N, Mengin J.
\newblock Logical deduction using the local computation framework.
\newblock In: European Conference on Symbolic and Quantitative Approaches to
  Reasoning and Uncertainty. Springer; 1999. p. 386--396.

\bibitem{zadrozny2018sheaf}
Zadrozny W, Garbayo L. A Sheaf Model of Contradictions and Disagreements.
  Preliminary Report and Discussion; 2018.
\newblock Eprint available at
  {\href{http://arxiv.org/abs/1801.09036}{arXiv:1801.09036 [cs.CL]}}.

\bibitem{Bell-thm}
Bell JS.
\newblock On the {E}instein-{P}odolsky-{R}osen paradox.
\newblock Physics. 1964;1(3):195--200.

\bibitem{CHSH}
Clauser JF, Horne MA, Shimony A, Holt RA.
\newblock Proposed experiment to test local hidden-variable theories.
\newblock Physical Review Letters. 1969 Oct;23(15):880--884.

\bibitem{Bell:Speakable}
Bell JS.
\newblock Speakable and unspeakable in quantum mechanics: Collected papers on
  quantum philosophy.
\newblock Cambridge University Press; 1987.

\bibitem{Fine82}
Fine A.
\newblock Hidden variables, joint probability, and the {B}ell inequalities.
\newblock Physical Review Letters. 1982 Feb;48(5):291--295.

\bibitem{Hardy92:nonlocality1}
Hardy L.
\newblock Quantum mechanics, local realistic theories, and {L}orentz-invariant
  realistic theories.
\newblock Physical Review Letters. 1992 May;68(20):2981--2984.

\bibitem{Hardy93:nonlocality2}
Hardy L.
\newblock Nonlocality for two particles without inequalities for almost all
  entangled states.
\newblock Physical Review Letters. 1993;71(11):1665--1668.

\bibitem{GHZ}
Greenberger DM, Horne MA, Zeilinger A.
\newblock Going beyond {B}ell's theorem.
\newblock In: Kafatos M, editor. {B}ell's theorem, quantum theory, and
  conceptions of the universe. vol.~37 of Fundamental Theories of Physics.
  Kluwer; 1989. p. 69--72.

\bibitem{GHSZ90}
Greenberger DM, Horne MA, Shimony A, Zeilinger A.
\newblock {B}ell's theorem without inequalities.
\newblock American Journal of Physics. 1990;58(12):1131--1143.

\bibitem{haenni2004ordered}
Haenni R.
\newblock Ordered valuation algebras: a generic framework for approximating
  inference.
\newblock International Journal of Approximate Reasoning. 2004;37(1):1--41.

\bibitem{CaruPhD2019}
Car\'{u} G.
\newblock Logical and Topological Contextuality in Quantum Mechanics and
  Beyond.
\newblock University of Oxford; 2019.

\bibitem{AbramskyCaru2019Upcoming}
Abramsky S, Car\`u G.
\newblock Generic inference algorithms for contextuality; 2019.
\newblock Forthcoming.

\bibitem{shenoy1992valuation}
Shenoy PP.
\newblock Valuation-based systems: A framework for managing uncertainty in
  expert systems.
\newblock In: Fuzzy logic for the management of uncertainty. John Wiley \&
  Sons, Inc.; 1992. p. 83--104.

\bibitem{MansfieldFritz11:Hardy}
Mansfield S, Fritz T.
\newblock {H}ardy's non-locality paradox and possibilistic conditions for
  non-locality.
\newblock Foundations of Physics. 2012;42:709--719.
\newblock Eprint available at
  {\href{http://arxiv.org/abs/1105.1819}{arXiv:1105.1819 [quant-ph]}}.

\end{thebibliography}
\end{document}